\documentclass[aps,superscriptaddress,twocolumn,twoside,floatfix,pra,a4paper]{revtex4-2} 

\usepackage{comment}
\usepackage{subfigure}
\usepackage{times}
\usepackage{epsfig}
\usepackage{amsfonts}
\usepackage{amsmath}
\usepackage{amssymb,amsthm}
\usepackage{xcolor,colortbl}
\usepackage{multirow}
\usepackage{braket}
\usepackage{latexsym}
\usepackage{tabularx}
\usepackage{amsfonts}
\usepackage{mathrsfs}
\usepackage{natbib}
\usepackage{verbatim}
\usepackage{gensymb}
\usepackage{caption}
\usepackage{caption}
\usepackage{subcaption}
\usepackage{ragged2e}
\usepackage{ulem}
\DeclareCaptionJustification{justified}{\justifying}
\usepackage{blkarray}
 \usepackage{graphicx}
 \usepackage[export]{adjustbox}
 \newcommand{\ba}{\begin{eqnarray}}
\newcommand{\ea}{\end{eqnarray}}
\newcommand{\ketbra}[2]{|#1\rangle \langle #2|}

\newtheorem{theorem}{Theorem}

\usepackage[colorlinks=true,linkcolor=blue,citecolor=magenta,urlcolor=blue]{hyperref}
\allowdisplaybreaks

\newtheorem{Lemma}{Lemma}

\newcommand{\sbg}{\textcolor{red}}

\begin{document}

\title{ Minimum Detection Efficiencies for Loophole-free Genuine Nonlocality Tests  }

\author{Subhendu B. Ghosh}
\affiliation{Physics and Applied Mathematics Unit, 203 B.T. Road Indian Statistical Institute Kolkata, 700108}

\author{Snehasish Roy Chowdhury}
\affiliation{Physics and Applied Mathematics Unit, 203 B.T. Road Indian Statistical Institute Kolkata, 700108}

\author{Ranendu Adhikary}
\affiliation{Electronics and Communication Sciences Unit, 203 B.T. Road Indian Statistical Institute Kolkata, 700108 }

\author{Arup Roy}
\affiliation{Department of Physics, A B N Seal College Cooch Behar, West Bengal 736101, India}

\author{Tamal Guha}
\affiliation{QICI Quantum Information and Computation Initiative, Department of Computer Science, The University of Hong Kong, Pokfulam Road, Hong Kong}

\begin{abstract}
The certification of quantum nonlocality, which has immense significance in architecting device-independent technologies, confronts severe experimental challenges. Detection loophole, originating from the unavailability of perfect detectors, is one of the major issues amongst them. In the present study we focus on the minimum detection efficiency (MDE) required to detect various forms of genuine nonlocality, originating from the type of causal constraints imposed on the involved parties. 
In this context, we demonstrate that the MDE needed to manifest the recently suggested $T_2$-type nonlocality deviates significantly from perfection. Additionally, we have computed the MDE necessary to manifest Svetlichny's nonlocality, with state-independent approach  markedly reducing the previously established bound. Finally, considering the inevitable existence of noise we demonstrate the robustness of the imperfect detectors to certify $T_2$-type nonlocality.

\end{abstract}

\maketitle
{\it Introduction.--}
Over the past thirty years, a distinctive category of information  theory has emerged popularly known as Device-Independent (DI) technology, which surpasses several limitations of classical physics \cite{Ekert1991,Scarani2012, Barrett2005a, Acin2006, Pironio2010, Colbeck2012, Chaturvedi2015, Mukherjee2015, Brunner2013, Pappa2015, Buhrman2010, Roy2016, Banik2019}. 
  Admittedly, the foundational constituent behind such an ever expanding architecture was implanted in the seminal results of John S. Bell \cite{Bell1964,Bell1966}, which derives a testable criteria for certifying the nonlocal aspect of any operational theory. In a nutshell, the nonlocality is a signature of multipartite correlations, incompatible with the classical outlook to the causal constraints imposed on the involved parties. For example, the pioneering nonlocality test as proposed by Clauser-Horne-Shimony-Holt (CHSH) considers a bipartite scenario, where spatially separated agents, supplemented with some noncommunicating correlation, are asked to produce a dichotomic random variable depending upon their individual classical inputs \cite{clauser1969}. Finally, an inequality based on their local outputs is so designed that the violation of the same led to an apparent contradiction with their imposed causal structure, which can further be removed by relaxing the local-realistic description of their shared correlation. This in turn certifies the nonlocal signature of that correlation. However in the multipartite scenario, there are various causal constraints that can be imposed on the spatially separated parties and accordingly with an increasing number of parties involved, the complexity associated with their local input-output statistics also grows. In the most basic multipartite scenario, three spatially separated parties are asked to generate local random variables depending upon the input given to them locally. On the other hand, when there is no causal constraint imposed on two of the parties (personified as Bob and Charlie), then their local input-output statistics should be factorized between Alice vs. Bob-Charlie when their measurements are space like separated from that of Alice. While any of their obtained statistics compatible with such a scenario is termed as bi-local (BL) correlations, an apparent contradiction of the same can be characterized by an inequality proposed by Svetlichny \cite{Svetlichny1987} and coined as genuinely nonlocal. In a more complicated scenario, one may further consider a temporal ordering in the input-output generation of the parties Bob and Charlie. Intuitively, this allows Charlie to make a redundant signalling to Bob, when he is in causal future of Bob and vice-versa. It is instructive that such a causal constraint is stronger than that of the earlier. Accordingly the set of correlations compatible with these constraints, namely the time-ordered bi-local (TOBL) correlations is strictly included inside the set of BL correlations \cite{Navascués2012}. In isolation, the characterization of BL correlations has emerged over the past decade when an operational inconsistency of the Svetlichny's definition of bi-locality is reported and it has further fuelled the development of various refined causal structures on genuine nonlocality \cite{Navascués2012,Pironio2013,Duttapra2020}.

Interestingly, quantum theory exhibits the signature of nonlocality in its simplest possible bipartite scenario. However, in the recent era of many body physics the importance of multipartite quantum correlations such as entanglement and nonlocality is needless to over-emphasize. Notably, the genuine nonlocality has garnered significant attention, primarily owing to its implications in Device-Independent Random Number Generation (DIRNG), Device-Independent Quantum Key Distribution (DIQKD), and Device-Independent certification of genuine entanglement \cite{Bancalprl2011,Chenprl2014,Acinquantum2018,Brussprr2020, Bancalbook2014, Liangprl2014, Aolitaprl2012}. However, the experimental certification of the potential nonlocal signature of quantum correlations encounter serious challenges by the possible loopholes in Bell-test. One such major operational loophole is the detection loophole, which emerges due to an inefficient photon detector, used to record the local input-output statistics obtained by each of the spatially separated agents. In particular, there are possible instances where the imperfect detector may yield no click, i.e., an inclusive outcome. While one might consider overcoming this issue by simply discarding the no-click outcomes, such post-selection requires an additional assumption, specifically the fair-sampling assumption. Without this assumption, even a local hidden variable model can violate Bell inequalities \cite{Perale1970}. A more comprehensive and appropriate approach to addressing this challenge is to regard the no-click event as a potential valid outcome achievable during the Bell test. The second approach needs a cut-off on the detection efficiency to observe nonlocality. However, given an entangled state, accompanied with a set of observable may require further bound on the minimum detection efficiency. For instance, Bell test involving two qubit maximally entangled state demands 83${\%}$ detection efficiency, while non maximally entangled state needs 67${\%}$ \cite{Garg1987, Larssson1998, Eberhard1993}. Moreover, motivated by experimental set-up, a minimum detector efficiency for the asymmetric bipartite Bell-test has also been reported \cite{Larson2007,Brunner2007prl}. 
 The study of such bipartite nonlocality tests under inefficient detectors has gained further importance by developing a series of sophisticated experimental set-up  \cite{Monroe2004,Hanson2015,Zeilinger2015,Shalm15, Wallraff2023}. In spite of having several foundational and practical importance, in the multipartite scenario the nonlocality test with inefficient detectors is mostly restricted to its simplest form \cite{Semitecolos2001,Villanueva2008}, with a recent addition for Svetlichny-type genuine nonlocality test \cite{Hong2012, Smerzi2012pra}. 
  
  On the other hand, the resource theoretic construction of quantum nonlocality contradicts the existence of an unique nonlocal measure, even in its simplest scenario \cite{Banik2023}. This further motivates to investigate possible detection efficiencies for other possible tests for genuine nonlocality by imposing different causal-constraints on the involved parties \cite{Navascués2012,Pironio2013,Duttapra2020}. In this work, we consider various possibilities of nonlocality tests and derive the minimal possible sophistication one may require for the detectors to certify genuine nonlocality in the proposed scenario. With the help of a state-independent generalized study, we have been able to significantly lower the already existing bound on detection efficiency, for the Svetlichny-type nonlocality test. Moreover, our results encapsulates an experimentally motivated scenario for hybrid entanglement, where different quantum particles are allowed to be entangled via inter-particle interactions \cite{Jeong2014, Wen2021, Feist2022, He2022}. Finally, a brief analysis of noise robustness for the inefficient detectors have been reported. This, in turn, hints towards a complimentary relation between the detection efficiency and the possible range of one-parameter quantum settings (i.e., the state along with the triplets of measurement pairs) to exhibit genuine nonlocality, in the presence of an unavoidable noise.

\textit{Genuine nonlocality under various causal constraints.}-- Consider the simplest multipartite scenario, where three spatially separated agents perform local operations on their individual subsystems. Further, their space-time relation is restricted in such a way that no more than two of them can reside in a same light cone. As a consequence, any local-realistic hidden variable description simulating their local input-output statistics must be factorized in either of the three possible bi-partitions. This in essence reads,
\begin{align}\label{svetlichnylocal}
&P(abc|XYZ) = \sum_\lambda q_{\lambda}\,P_\lambda(ab|XY)\,P_\lambda(c|Z) + \nonumber\\
& \sum_\mu q_{\mu}\,P_\mu(ac|XZ)\,P_\mu(b|Y) + \sum_\nu q_{\nu}\,P_\nu(bc|YZ)\,P_\nu(a|X),
\end{align}
where, $0\leq q_{\lambda},q_{\mu},q_{\nu} \leq 1$ are the respective probabilities compatible with the causal constraint and $\sum_\lambda q_{\lambda}+\sum_\mu q_{\mu}+\sum_\nu q_{\nu}=1$. In 1987, Svetlichny pointed out that all possible correlations that can be decomposed in the form of Eq.(\ref{svetlichnylocal}), respects the following inequality over their operational statistics \cite{Svetlichny1987}:
\begin{align}\label{ineqSI}
    -\langle A_{0}B_{0}C_{0}\rangle-\langle A_{1}B_{0}C_{0}\rangle+\langle A_{0}B_{1}C_{0}\rangle-\langle A_{1}B_{1}C_{0}\rangle \nonumber\\ 
    -\langle A_{0}B_{0}C_{1}\rangle+\langle A_{1}B_{0}C_{1}\rangle-\langle A_{0}B_{1}C_{1}\rangle-\langle A_{1}B_{1}C_{1}\rangle \le 4, 
\end{align}
where $\langle A_xB_yC_z\rangle$ for every possible $ x,y,z\in\{0,1\}$ denotes the expectation value of the two-outcome $\{\pm 1\}$ measurements performed by the respective parties on their local constituents. In contrast, any correlation violating the above inequality (\ref{ineqSI}) can be characterized as Svetlichny-type genuine nonlocal.

It is instructive that for the above scenario the parties, sharing a common light cone (say Bob and Charlie), are free to communicate with each other. However, such a relaxed causal constraint allows them to establish a bipartite nonlocality with Alice. With this operational inconsistency, one may further impose a definite causal order between input-output generation of Bob and Charlie, which, in effect, allows only a single direction of relevant communication between them \cite{Navascués2012,Pironio2013}. Any correlation compatible with such a constraint is referred to be $T_{2}$ local and admits the following decomposition:
\begin{align}\nonumber\label{TO} &P(abc|xyz) = \sum_{\lambda}q_{\lambda}P_{\lambda}^{T_{AB}}(ab|xy) P_\lambda (c|z) + \nonumber\\&  \sum_{\mu}q_{\mu}P_{\mu}^{T_{AC}}(ac|xz) P_\mu (b|y)   +\sum_{\nu} q_{\nu}P_{\nu}^{T_{BC}}(bc|yz) P_\nu (a|x) \end{align}
where, the joint term $p_{\lambda}^{T_{AB}}(ab|xy)=p_{\lambda}^{A<B}(ab|xy) $  if the local statistics obtained by Alice is in causal past of Bob and $p_{\lambda}^{T_{AB}}(ab|xy)=p_{\lambda}^{B<A}(ab|xy) $ when the order is reversed. All the other bipartite marginals can also be explained in similar fashion. Notably, in both the cases there are (at most) 1-way signaling. This refinement seem to offer a potential solution to sidestep the earlier mentioned issues that might arise within the Svetlichny scenario. Now if a probability distribution $p(abc|xyz)$ can be written in the form \eqref{TO} then it is said to be $T_{2}$ local otherwise it is called genuinely three-way nonlocal. In \cite{Pironio2013}, the authors proposed an experimentally testable criteria to verify the compatibility of a tripartite correlation with such a time-ordered space-time structure:
{\small
  \setlength{\abovedisplayskip}{5pt}
  \setlength{\belowdisplayskip}{\abovedisplayskip}
  \setlength{\abovedisplayshortskip}{0pt}
  \setlength{\belowdisplayshortskip}{10pt}

\begin{align} 
\mathcal{B}_{T_2} :=& -2\{P(00|A_1B_1) + P(00|B_1C_1)+P(00|A_1C_1)\}&& \nonumber \\
& - P(000|A_0B_0C_1) - P(000|A_0B_1C_0)&&\nonumber\\
& - P(000|A_1B_0C_0)+ 2P(000|A_1B_1C_0)&& \nonumber  \\ 
&+ 2P(000|A_1B_0C_1) + 2P(000|A_0B_1C_1)&& \nonumber  \\
& + 2P(000|A_1B_1C_1) \leq 0,&&  \label{ineq}
\end{align}
}
where, $A_x,B_y,C_z$ are the dichotomic local observable with possible outcomes $\{0,1\}$. A violation of the inequality (\ref{ineq}) certifies the genuine nonlocality possessed by the correlation.


 {\it Minimum detection efficiency (MDE)--} The requirement of a loophole-free experimental test for nonlocality can be classified into two different aspects: Firstly, from the information theoretic perspective, it is required to obtain a nonlocal statistics out of a black box promising to generate a particular form of correlation or state. On the other hand, from the foundational perspective, one may eager to certify the nonlocal signature of an operational theory in a device-independent manner. Notably, while the first one conceives a state-dependent outlook, the latter one asks about the nature of the theory regardless of any possible preparations. These two different perspectives demand different forms of experimental sophistication, one of which is the efficiency of the detectors. However, in the following, we will see how they are connected via an optimization.\par
 Let us, consider one such Bell-inequality $\mathcal{B}$ for nonlocality test, which is supposed to be verified on an $N$-partite entangled state $\rho$ with $i^{\text{th}}$ party performing $m_i$ numbers of incompatible measurements $\mathcal{M}_i:=\{M_1,\cdots, M_{m_i}\}$. Then the detection efficiency required for such a test is said to be the cut-off detection efficiency (CDE) and is denoted as
 \begin{align}\label{cutoff}
\eta^\mathcal{B} := \eta^\mathcal{B}(\rho,\{\mathcal{M}_k\}_{k=1}^{N}).   
\end{align}
However, the minimum detection efficiency (MDE) for the same Bell-test $\mathcal{B}$ can be identified as, 
\begin{align}\label{eq1}
    \eta^\mathcal{B}_{\min} = \underset{\rho,\{\mathcal{M}_k\}_{k=1}^{N}}{\text{inf}}\eta^\mathcal{B}(\rho,\{\mathcal{M}_k\}_{k=1}^{N}),    
\end{align}
where the infimum is over the set of all possible $N$-partite entangled density matrices acting on the joint Hilbert space $\mathcal{H}$ and $\{\mathcal{M}_k\}_{k=1}^N$ are the set of all possible $m_k$ incompatible measurements for the $k^{\text{th}}$-cite. Notably, CDE is a quantity of interest for the self-testing scenario, while MDE can be assigned with the device-independent certification of nonlocality in quantum theory.

Generally, the CDE is a very complicated nonlinear function of both the states and measurements. Hence, the optimization involved in (\ref{eq1}) is highly nontrivial and the characterization of MDE for an arbitrary Bell-test becomes very challenging in general. However, motivated by the device-independent architectures and their implications, in the present work we derive the MDE for detection of various tripartite genuine nonlocality. In particular, for the $T_2$-type nonlocality test we estimate $\eta_{\min}$ bypassing the optimization complexities related to CDE, while for the Svetlichny-type test, we adopt a rigorous numerical optimization.

Furthermore, in a real-life experimental setup, things become more intricate due to the inevitable existence of noise. 
If, by $\{\mathbf{p}\}$ one denotes the set of parameters defining the noise, eq.(\ref{eq1}) is modified as:
\begin{equation}\label{eq2}
    \eta^\mathcal{B}_{\min}(\{\mathbf{p}\}) = \underset{\rho,\{\mathcal{M}_k\}_{k=1}^{N}}{\text{inf}}\eta^\mathcal{B}(\rho,\{\mathcal{M}_k\}_{k=1}^{N},\{\mathbf{p}\}),    
\end{equation}
Clearly, in the laboratory the efficiencies of the detectors must follow the relation: $\eta > \eta^\mathcal{B}_{\min}(\{\mathbf{p}\})$.\par
\textit{MDE for genuine nonlocality.}-- Let us begin with a stronger causal constraint where the parties residing in the same light cone are bounded with a definite causal order to generate their local input-output statistics. As mentioned earlier, the incompatibility of a correlation in such a scenario can be certified by the violation of the inequality (\ref{ineq}). To establish such a genuine three-way nonlocal signature of shared quantum correlation, we consider a situation where the associated parties will perform the local measurements with inefficient detectors with corresponding efficiencies $\eta_A,~\eta_B\text{ and }\eta_C$. 
The following lemma derives the minimum requirement for detection efficiencies to violate (\ref{ineq}).  
\begin{Lemma} \label{lem1}
The spatially separated parties would be able to certify genuine three-way nonlocality in terms of the inequality $\mathcal{B}_{T_2}$, when the following condition holds
\begin{align*}
    ( 4\eta_{A}\eta_{B}\eta_{C}-\eta_{A}\eta_{B}-\eta_{A}\eta_{C}-\eta_{B}\eta_{C})>0.    
\end{align*}

\end{Lemma}

\begin{proof}
   From now onward, we will use $P^{\prime}(\mathcal{O}|\mathcal{X})$ and $P\left(\mathcal{O}|\mathcal{X}\right)$ for theoretical and observed probability respectively, when input variable $\mathcal{X}$ produces output variable $\mathcal{O}$. Assuming the independence of each of the individual local detectors the observed probability for each individual can only depend upon the corresponding theoretical probability and the detection efficiency of the local detector. Now, if all the parties agrees to assume $\mathcal{O}\equiv 1$ for the no-click instances, then $P(\mathcal{O}|{\mathcal{X}}) = \eta P^{\prime}\left(\mathcal{O}|\mathcal{X}\right)$, only when $\mathcal{O}\equiv 0$. This readily implies,
{\small\begin{align}
P(000|A_1B_1C_1)&=\eta_{A}\eta_{B}\eta_{C}P^{\prime}(000|A_1B_1C_1)&&\nonumber\\ &\leq\eta_{A}\eta_{B}\eta_{C}~~ \underset{X}{\min} ~ P^{\prime} (00|X)&&\nonumber
\end{align}}
where, $X\in\{A_1B_1,B_1C_1,A_1C_1\}$ .
{\small\begin{align}
\text{Again, }
&P(000|A_1B_1C_0)-P(00|A_1B_1)&&\nonumber\\
&= \eta_{A}\eta_{B}\eta_{C}P^{\prime}(000|A_1B_1C_0)-\eta_{A}\eta_{B}P^{\prime}(00|A_1B_1)&&\nonumber\\
&\leq\eta_{A}\eta_{B}\eta_{C}\underset{X}{\min}P^{\prime} (00|X) -\eta_{A}\eta_{B}\underset{X}{\min}P^{\prime}(00|X)&&\nonumber 
\end{align}}
where the last inequality holds since $P(000|A_1B_1C_0)-P(00|A_1B_1)\leq 0$ and $\eta_{C}, \eta_{A}, \eta_{B}\le1$. 
Deriving the similar inequalities for other pair of terms $P(000|A_1B_0C_1)-P(00|A_1C_1)$ and $P(000|A_0B_1C_1)-P(00|B_1C_1)$, we can rewrite the violation of $\mathcal{B}_{T_2}$ (\ref{ineq}) as
{\small\begin{align}
&( 4\eta_{A}\eta_{B}\eta_{C}-\eta_{A}\eta_{B}-\eta_{A}\eta_{C}-\eta_{B}\eta_{C})~~ \underset{X}{\min} ~ P^{\prime} (00|X)>0\nonumber
\\&\Rightarrow( 4\eta_{A}\eta_{B}\eta_{C}-\eta_{A}\eta_{B}-\eta_{A}\eta_{C}-\eta_{B}\eta_{C})>0\label{necT2}
\end{align}}
Hence, the above condition among the detector efficiencies, becomes necessary to exhibit genuine three-way nonlocality.
\end{proof}
At this point, a pertinent question naturally arises regarding the sufficiency of the inequality (\ref{necT2}). In other words, whether every triplet of $\{\eta_A,\eta_B,\eta_C\}$, satisfying $ 4\eta_{A}\eta_{B}\eta_{C}-\eta_{A}\eta_{B}-\eta_{A}\eta_{C}-\eta_{B}\eta_{C}>0$, corresponds to CDE for $\mathcal{B}_{T_2}$. In the following, we will answer this question affirmatively.

\begin{Lemma} \label{lem2}
There exists a quantum settings with a single parameter $\theta$, which violates the inequality (\ref{ineq}), namely $\mathcal{B}_{T_2}$, whenever $( 4\eta_{A}\eta_{B}\eta_{C}-\eta_{A}\eta_{B}-\eta_{A}\eta_{C}-\eta_{B}\eta_{C})>0$.
\end{Lemma}
\noindent
By the single parameter quantum settings, we would like to mean a class of genuinely entangled states along with two incompatible dichotomic measurements for each individual party, both parameterized by a single parameter $\theta$.
\begin{proof}
   Consider the three-qubit state
{\small\begin{align}\label{state}
    \ket{\Psi_{ABC}(\theta)}=k_{\theta}[(\ket{011}+\ket{101}+\ket{110})+\frac{(1-3\cos\theta)}{\sin\theta}\ket{111}]
\end{align}}
where $k_{\theta}$ is a function of $\theta$, denoting the normalization of the state. Each party performs an identical set of incompatible measurements in the following bases,
{\small\begin{align*}
    &X_{0}\equiv\{\ket{0},\ket{1}\}\\
    &X_{1}\equiv\{\cos\theta\ket{0}+\sin\theta\ket{1},\sin\theta\ket{0}-\cos\theta\ket{1}\}
\end{align*}}
where $X\in\{A,B,C\}$. Such a measurement setting readily implies  
$$P(000|A_iB_jC_k)=0$$
where, at least two of $\{i,j,k\}$ equals to $0$. On the other hand when two of $\{i,j,k\}$ equals to $1$, probabilities reduce to,
\begin{align*}
P(000|A_iB_jC_k)=\eta_{A}\eta_{B}\eta_{C}k_{\theta}^{2}\sin^{4}{\theta},\end{align*}

 and the following marginals become
 \begin{align*}P(00|X_1Y_1)=\eta_{X}\eta_{Y}k_{\theta}^{2}\sin^{4}{\theta}(1+\tan^{2}{\frac{\theta}{2}}),
 \end{align*}
 where, $X\neq Y\in\{A,B,C\}$. 

\par
Replacing the above expressions for the imperfect detectors, the inequality (\ref{ineq}) reduces to the following
\begin{multline}
    k_{\theta}^{2}\sin^{4}{\theta}[4\eta_{A}\eta_{B}\eta_{C}-\eta_{A}\eta_{B}-\eta_{B}\eta_{C}-\eta_{A}\eta_{C}\\-\tan^{2}{\frac{\theta}{2}}(\eta_{A}\eta_{B}+\eta_{B}\eta_{C}+\eta_{A}\eta_{C})] \label{ineq2} \leq 0
\end{multline}
Therefore, to violate the above inequality, we must have,
\begin{equation*}
    \tan^{2}{\frac{\theta}{2}} < \frac{4\eta_{A}\eta_{B}\eta_{C}-\eta_{A}\eta_{B}-\eta_{B}\eta_{C}-\eta_{A}\eta_{C}}{\eta_{A}\eta_{B}+\eta_{B}\eta_{C}+\eta_{A}\eta_{C}}
\end{equation*}
Therefore, for every possible value of $\{\eta_A,\eta_B,\eta_C\}$, satisfying the condition $4\eta_{A}\eta_{B}\eta_{C}-\eta_{A}\eta_{B}-\eta_{B}\eta_{C}-\eta_{A}\eta_{C}>0$ it is possible to find a range of $\theta$ for which a three-way genuine nonlocality can be certified experimentally. Notably, the inequality on the detection efficiencies is strict, otherwise one may consider $\theta=n\pi,~n\in \mathbb{Z}$ and the measurements $X_0$ and $X_1$ will become compatible.
\end{proof}

From Lemma  (\ref{lem1}) and Lemma  (\ref{lem2}), we can immediately infer the subsequent theorem. 
\begin{theorem}\label{t1}
   The MDE's for $\mathcal{B}_{T_2}$ for each of the party satisfy the following relation
   \begin{align*}
       &4\Big(\eta^{\mathcal{B}_{T_{2}}}_{\min}\Big)_A \Big(\eta^{\mathcal{B}_{T_{2}}}_{\min}\Big)_B \Big(\eta^{\mathcal{B}_{T_{2}}}_{\min}\Big)_C- \Big(\eta^{\mathcal{B}_{T_{2}}}_{\min}\Big)_A \Big(\eta^{\mathcal{B}_{T_{2}}}_{\min}\Big)_B \\
       &- \Big(\eta^{\mathcal{B}_{T_{2}}}_{\min}\Big)_A \Big(\eta^{\mathcal{B}_{T_{2}}}_{\min}\Big)_C- \Big(\eta^{\mathcal{B}_{T_{2}}}_{\min}\Big)_B \Big(\eta^{\mathcal{B}_{T_{2}}}_{\min}\Big)_C=0.
    \end{align*}
\end{theorem}
As a consequence of Theorem \ref{t1} the MDE for two different scenarios can be concluded. 
 For the case of symmetric detection efficiencies, i.e. $\eta_{A}=\eta_{B}=\eta_{C}=\eta$, the inequality (\ref{ineq}) can be violated if and only if $\eta>\eta^{\mathcal{B}_{T_{2}}}_{\min}=75\%$. On the other hand, when a hybrid entanglement between different quantum particles are observed, then the MDE for each of the particles may demand different experimental sophistication. In such a scenario, if two of the detectors work perfectly, i.e., say $\eta_{A}=\eta_{B}=100\%$, then
the inequality (\ref{ineq}) can be violated if and only if $\eta_{C}>\Big(\eta^{\mathcal{B}_{T_{2}}}_{\min}\Big)_C=50\%$. Therefore, perfect detectors are not imperative in this scenario and the genuine non-locality can be established even when detectors deviate significantly from being perfect.\par

{\it Background Noise.}--  While the above analysis is restricted to pure entangled states, in the practical experimental set-up the existence of background noise becomes unavoidable. The presence of noise in a preparation device is generally regarded as a junk state, prepared with a significantly low frequency. Without having any prior knowledge about the noise, a randomization over all possible junk states can be assumed as noise. In presence of such a \textit{white noise} the tripartite state takes the form 
\begin{align*}
    {\rho}_{ABC}=(1-p)\ketbra{{\Psi}(\theta)}{{\Psi}(\theta)}_{ABC}+\frac{p}{8}\mathbb{I},
\end{align*}
where, the state $\ket{\psi(\theta)}$ is same as in Eq. (\ref{state}). Naturally, the MDE $\eta_{\min}^{\mathcal{B}_{T2}}(\theta, p)$ (as introduced in Eq. (\ref{eq2})) becomes highly nonlinear function. However, we show that with a reasonable amount of detection efficiency $0.92\leq\eta^{\mathcal{B}_{T2}}\leq0.96$, one can expect tolerance even up to $1.6\%$ background noise (see Fig. \ref{fig1} for details).\par

\begin{figure*}[tbh!]
  \centering
  \subfigure[]{\includegraphics[width = 0.474\textwidth]{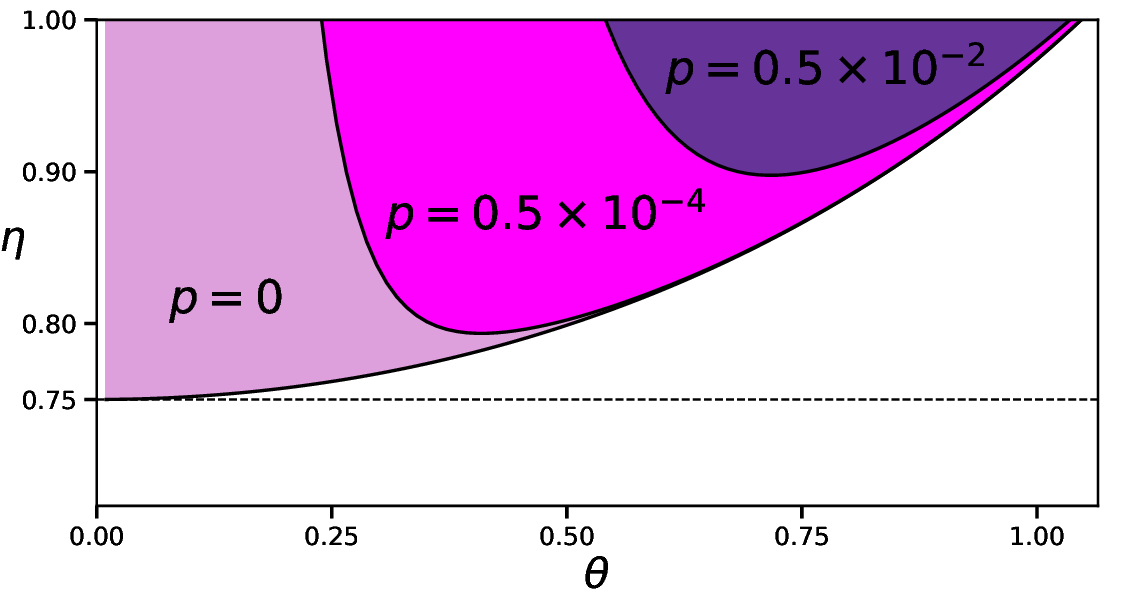}\label{fig1a}}\quad
  \subfigure[]{\includegraphics[width = 0.507\textwidth]{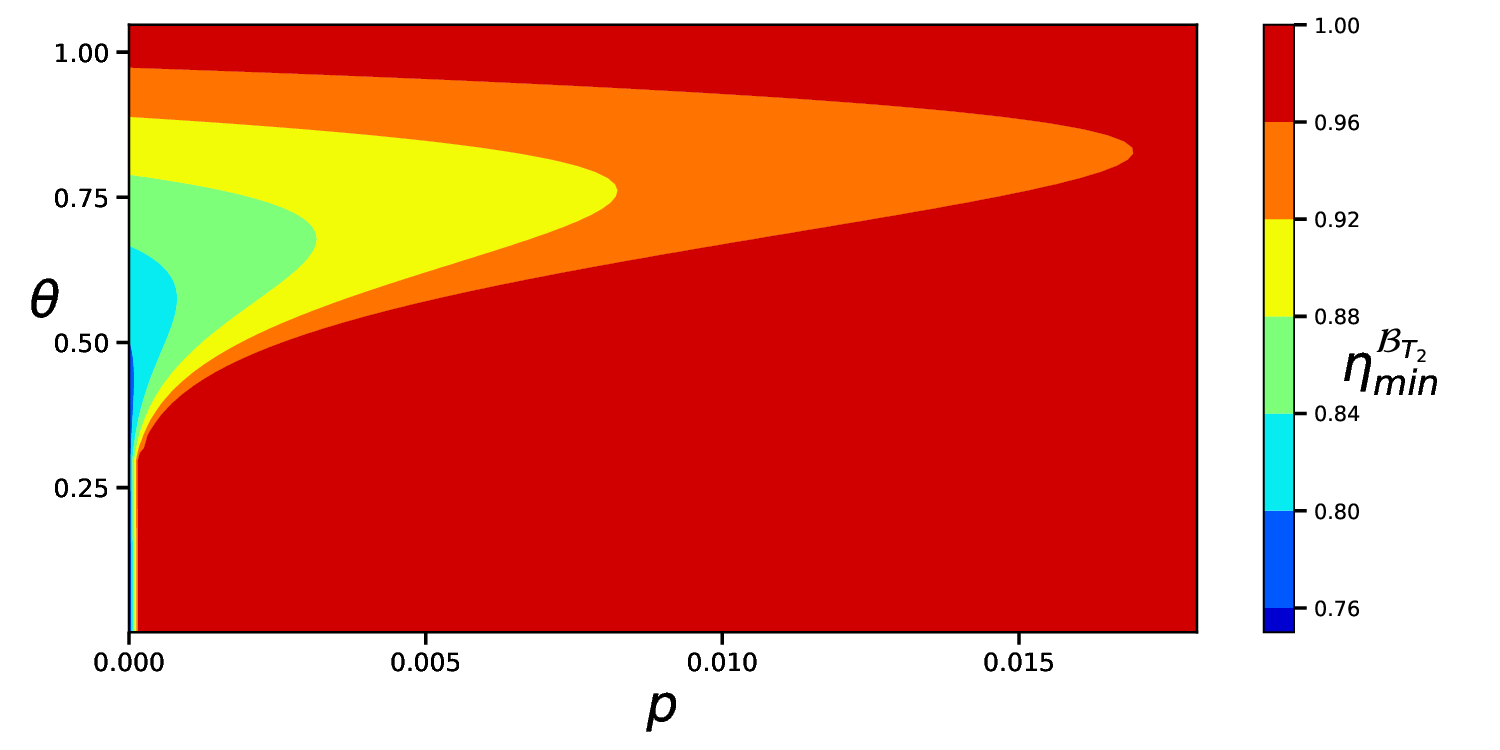}\label{fig1b}}
  \caption{\ref{fig1a} contains $\eta$ vs $\theta$ plot for different $p$. The solid black curve represents MDE, i.e, $\eta_{min}^{\mathcal{B}_{T_2}}$. The detector should have detection efficiency within the shaded(color for online) region in order to exhibit $T_2$ nonlocality. \ref{fig1b} contains three parameter plot of $p$ vs $\theta$ vs $\eta_{min}^{\mathcal{B}_{T_2}}$. It shows that $\mathcal{B}_{T_{2}}$ is fairly noise-tolerent(upto roughly $1.6\%$).}
    \label{fig1}
\end{figure*}

\textit{MDE for Svetlichny Nonlocality.}-- We will now consider a genuine nonlocality experiment with relaxed causal constraints on involved parties, i.e., by invoking no-time orderings among their local input-output generation. The violation of Svetlichny inequality (\ref{ineqSI}) serves as a prominent benchmark for illustrating any incompatibility with such a causal-structure and is identified as a Svetlichny-type genuinely nonlocal correlation. However, the experimental sophistication required to certify such a violation is pretty high -- almost 97\% of CDE for three-particle GHZ state \cite{Hong2012}. In the following, we will first derive an expression for CDE, depending on the statistics obtained by the individual parties. Thereafter, we numerically estimate the MDE considering any possible tripartite entangled preparation, which is significantly lower than the bound obtained in \cite{Hong2012}.


Let us first express the inequality (\ref{ineqSI}) in terms of the outcomes probabilities. Relabeling the outcomes of each measurement $\{\pm 1\}\mapsto\{0,1\}$, we can rewrite the inequality as,
\begin{align}\label{inSI2}
   &4+8[m_{010}+m_{101}-(m_{000}+m_{001}+m_{011}+m_{100}+m_{110}&&\nonumber\\
   &+m_{111})]+8[q_{00}+q_{11}+r_{01}+r_{10}+s_{00}+s_{11}]&&\nonumber\\
   &-4[a_{0}+a_{1}+b_{0}+b_{1}+c_{0}+c_{1}]\leq4&&
\end{align}
\begin{align*}
&\text{where, } m_{ijk}=P(000|A_iB_jC_k),q_{ij}=P(00|A_iB_j),&&\\
&r_{jk}=P(00|B_jC_k), s_{ik}=P(00|A_iC_k),a_i=P(0|A_i),&&\\
&b_j=P(0|B_j), c_{k}=P(0|C_k),~\forall i,j,k \in \{0,1\}.
\end{align*}

\begin{theorem}\label{t2}
     For the Svetlichny-type nonlocality test, with symmetric detectors, the CDE becomes $\eta^{\text{S}}=\frac{\sqrt{\beta^2+4|\alpha|\gamma}-\mathbf{sgn}(\alpha)\times\beta}{2|\alpha|}$, where
     {\small\begin{align*}
         &\beta:=2\sum_{i,j,k}[\overline{(i\oplus j)}q_{ij}+(j\oplus k)r_{ij}+\overline{(i\oplus k)}s_{ik}],\\
         &\alpha:=2\sum_{i,j,k}(-1)^{(\bar{i}.\bar{j}+j.k+i.\bar{k})}m_{ijk},~\gamma:=\sum_{i,j,k}(a_i+b_j+c_k)\\
         & \text{ and } \mathbf{sgn}(\alpha)=\pm 1,\text{ for }\alpha>0 \text{ and }\alpha<0 \text{ respectively}.
     \end{align*}}
\end{theorem}
\begin{proof}
    The Svetlichny's inequality, in the form (\ref{inSI2}), can be violated with a properly chosen quantum setting, whenever
    \begin{align*}
   &8[m_{010}+m_{101}-(m_{000}+m_{001}+m_{011}+m_{100}+m_{110}&&\nonumber\\
   &+m_{111})]+8[q_{00}+q_{11}+r_{01}+r_{10}+s_{00}+s_{11}]&&\nonumber\\
   &-4[a_{0}+a_{1}+b_{0}+b_{1}+c_{0}+c_{1}]>0&&
    \end{align*}
Now, consider the situation where all three imperfect detectors are of the same detection efficiency $\eta$. Then replacing all the theoretical probabilities with the observed probabilities, we can rewrite the above inequality as,
{
\small
\begin{align*}
&\sum_{i,j,k~\in\{0,1\}}\bigg[(-1)^{(\bar{i}.\bar{j}+j.k+i.\bar{k})}2\eta^3m_{ijk}+2\eta^2\Big\{\overline{(i\oplus j)}q_{ij}&&\nonumber\\
&({j\oplus k})r_{jk}+\overline{(i\oplus k)}s_{ik}\Big\}-\eta\Big(a_i+b_j+c_k\Big)\bigg]>0&&
\end{align*}
}
The detection efficiency being always positive, above inequality simplifies to,
\begin{equation*}
    \alpha \eta^2+\beta \eta-\gamma>0
\end{equation*}
where, $\alpha,~\beta\text{ and }\gamma$ are same as defined earlier.\par
Note that, for any quantum settings that exhibit Svetlichny-type nonlocality, must satisfy the above inequality at least for $\eta=1$, i.e., for the perfect detector. Hence $\alpha+\beta+\gamma>0$. This, together with the fact that $\beta\geq0$ and $\gamma\geq0$, provides $\eta=\frac{\sqrt{\beta^2+4|\alpha|\gamma}-\mathbf{sgn}(\alpha)\times\beta}{2|\alpha|}$ saturating the above inequality. Note that, the $\mathbf{sgn}(\alpha)$ function incorporates that $\alpha$ can be both positive or, negative. Finally, noting that all the parameters $\alpha,\beta,\gamma$ are the parameters generated from the quantum settings, i.e., state and measurements, we can identify $\eta$ as $\eta^{\text{S}}$ -- the CDE for Svetlichny-type nonlocality test.
\end{proof}
Optimizing over all possible quantum settings numerically, it can be shown that the $\eta_{\text{min}}^{S}\simeq 88.1\%$ is sufficient for demonstrating Svetlichny-type nonlocality.

{\it Discussion--} In summary, we have estimated the minimum detection efficiency required for loophole free test of various forms of genuine nonlocality exhibited by multipartite quantum systems. In particular, taking the operational framework into account we have investigated both the time-ordered ($T_2$) genuine nonlocality, as well as the traditional Svetlichny-type genuine nonlocality tests. To this goal, we have primarily introduced the notion of cut-off detection efficiency which explicitly depends upon the quantum settings used to demonstrate a particular Bell nonlocality. Then systematically we have reached to the minimum detection efficiency for the concerned nonlocality experiment, optimizing over all possible quantum states and measurement settings. Within this framework, our study reveals that, 
to demonstrate loophole-free $T_{2}$ nonlocality in symmetric cases, a detection efficiency of $75\%$ at each site is both necessary and sufficient. This requirement can be relaxed to $50\%$ when any two out of the three parties possess perfect detection efficiency. 
Interestingly, our results suggest that to achieve the MDE for violating $T_{2}$ inequality is more feasible in the scenarios where the correlations nearly breach the inequality, instead of demonstrating the maximal violations. Further, we have shown that in presence of a significantly low white-noise, the range of quantum settings showing $T_2$ nonlocality sharply decreases. Nevertheless, for the given inefficient detectors and a permissible noise limit, our analysis characterizes the range of quantum settings viable to exhibit genuine nonlocality in experiments. On the other hand, we have also established that a detection efficiency of $88.1\%$ is deemed sufficient for demonstrating Svetlichny nonlocality, which is well below the previously estimated one.\par %
Though Svetlichny inequality and $T_{2}$ inequality, both are distinct facets of the set of the $T_{2}$ correlations, the set of all $T_2$ local correlations is a strict subset of that of the Svetlichny-local correlations. This, in turn, admits that the correlations violating Svetlichny inequality are stronger than those violating the $T_2$ inequality and hence demand more sophistication in the involved experimental set-up. Thus, as a consequence of our results, the investigation and comparison of MDE across distinct nonlocal classes has emerged as a compelling avenue for future research. \\
\\

\textit{Acknowledgements.}--SRC acknowledges support from University Grants Commission, India (reference no. 211610113404). RA acknowledges funding and support from ISI DCSW Project No. PU/506/PL-MISC/521. We acknowledge fruitful discussions with Guruprasad Kar, Manik Banik, Asutosh Rai and Amit Mukherjee.


\begin{thebibliography}{1}


\bibitem{Scarani2012} V. Scarani; The device-independent outlook on quantum physics (lecture notes on the power of Bell's theorem), 
\href{http://www.physics.sk/aps/pub.php?y=2012&pub=aps-12-04}{Acta Physica Slovaca {\bf 62}, 347 (2012)}.

\bibitem{Ekert1991} A. K. Ekert; Quantum cryptography based on Bell’s theorem,
\href{https://doi.org/10.1103/PhysRevLett.67.661}{Phys. Rev. Lett. {\bf 67}, 661 (1991)}.

\bibitem{Barrett2005a} J. Barrett, L. Hardy, and A. Kent; No Signaling and Quantum Key Distribution,
\href{https://doi.org/10.1103/PhysRevLett.95.010503}{Phys. Rev. Lett. {\bf 95}, 010503 (2005)}.

\bibitem{Acin2006} A. Acín, N. Gisin, and L. Masanes; From Bell’s Theorem to Secure Quantum Key Distribution,
\href{https://doi.org/10.1103/PhysRevLett.97.120405}{Phys. Rev. Lett. {\bf 97}, 120405 (2006)}.

\bibitem{Pironio2010} S. Pironio, A. Acín, S. Massar, A. Boyer de la Giroday, D. N. Matsukevich, P. Maunz, S. Olmschenk, D. Hayes, L. Luo, T. A. Manning, and C. Monroe; Random numbers certified by Bell’s theorem,
\href{https://doi.org/10.1038/nature09008}{Nature {\bf 464}, 1021 (2010)}. 

\bibitem{Colbeck2012} R. Colbeck and R. Renner; Free randomness can be amplified, 
\href{https://doi.org/10.1038/nphys2300}{Nature Phys {\bf 8}, 450 (2012)}. 

\bibitem{Chaturvedi2015} A. Chaturvedi and M. Banik; Measurement-device-independent randomness from local entangled states, 
\href{https://doi.org/10.1209/0295-5075/112/30003}{Europhys. Lett. {\bf 112}, 30003 (2015)}.


\bibitem{Mukherjee2015} A. Mukherjee, A. Roy, S. S. Bhattacharya, S. Das, Md. R. Gazi, and M. Banik; Hardy's test as a device-independent dimension witness, 
\href{https://doi.org/10.1103/PhysRevA.92.022302}{Phys. Rev. A {\bf 92}, 022302 (2015)}.


\bibitem{Brunner2013} N. Brunner and N. Linden; Connection between Bell nonlocality and Bayesian game theory,
\href{https://doi.org/10.1038/ncomms3057}{Nat. Commun. {\bf 4}, 2057 (2013)}.

\bibitem{Pappa2015} A. Pappa, N. Kumar, T. Lawson, M. Santha, S. Zhang, E. Diamanti, and I. Kerenidis; Nonlocality and Conflicting Interest Games,
\href{https://doi.org/10.1103/PhysRevLett.114.020401}{Phys. Rev. Lett. {\bf 114}, 020401 (2015)}.
\bibitem{Buhrman2010} H. Buhrman, R. Cleve, S. Massar, and R. de Wolf; Nonlocality and communication complexity; \href{https://journals.aps.org/rmp/abstract/10.1103/RevModPhys.82.665}{Rev. Mod. Phys. \textbf{82}, 665 (2010)}
\bibitem{Roy2016} A. Roy, A. Mukherjee, T. Guha, S. Ghosh, S. S. Bhattacharya, and M. Banik; Nonlocal correlations: Fair and unfair strategies in Bayesian games,
\href{https://doi.org/10.1103/PhysRevA.94.032120}{Phys. Rev. A {\bf 94}, 032120 (2016)}.

\bibitem{Banik2019} M. Banik, S. S. Bhattacharya, N. Ganguly, T. Guha, A. Mukherjee, A. Rai, and A. Roy; Two-Qubit Pure Entanglement as Optimal Social Welfare Resource in Bayesian Game,
\href{https://doi.org/10.22331/q-2019-09-09-185}{Quantum {\bf 3}, 185 (2019)}.
\bibitem{Bell1964}  J. S. Bell; On the Einstein Podolsky Rosen paradox,
\href{https://doi.org/10.1103/PhysicsPhysiqueFizika.1.195}{Physics Physique Fizika \textbf{1}, 195 (1964)}.

\bibitem{Bell1966} J. S. Bell; On the Problem of Hidden Variables in Quantum Mechanics,
\href{https://doi.org/10.1103/RevModPhys.38.447}{Rev. Mod. Phys. {\bf 38}, 447 (1966)}.
\bibitem{clauser1969} J. F. Clauser, M. A. Horne, A. Shimony, and R. A. Holt,
\href{https://journals.aps.org/prl/abstract/10.1103/PhysRevLett.23.880}{Phys. Rev. Lett. \textbf{23}, 880 (1969)}.

\bibitem{Svetlichny1987} G. Svetlichny; Distinguishing three-body from two-body nonseparability by a Bell-type inequality, \href{https://journals.aps.org/prd/abstract/10.1103/PhysRevD.35.3066 (1987)}{Phys. Rev. D \textbf{35} 3066 (1987)}

\bibitem{Navascués2012} R. Gallego, L. E. Würflinger, A. Acín, and M. Navascués; Operational Framework for Nonlocality, \href{https://journals.aps.org/prl/abstract/10.1103/PhysRevLett.109.070401}{Phys. Rev. Lett. \textbf{109}, 070401 (2012)}

\bibitem{Pironio2013} J-D Bancal, J. Barrett, N. Gisin, and S. Pironio; Definitions of multipartite nonlocality, \href{https://journals.aps.org/pra/abstract/10.1103/PhysRevA.88.014102}{Phys. Rev. A \textbf{88}, 014102 (2013)}

\bibitem{Duttapra2020}  Sagnik Dutta, Amit Mukherjee and Manik Banik; Operational characterization of multipartite nonlocal correlations, \href{https://journals.aps.org/pra/abstract/10.1103/PhysRevA.102.052218}{Phys. Rev. A  \textbf{102}, 052218(2020)}

\bibitem{Chenprl2014} Qing Chen, Sixia Yu, Chengjie Zhang, C.H. Lai, and C.H. Oh; Test of Genuine Multipartite Nonlocality without Inequalities; \href{https://journals.aps.org/prl/abstract/10.1103/PhysRevLett.112.140404}{Phys. Rev. Lett. \textbf{112}, 140404}

\bibitem{Brussprr2020} Timo Holz, Hermann Kampermann, and Dagmar Bruß; Genuine multipartite Bell inequality for device-independent conference key agreement; \href{https://journals.aps.org/prresearch/abstract/10.1103/PhysRevResearch.2.023251}{Phys. Rev. Research  \textbf{2}, 023251(2020)}

\bibitem{Bancalprl2011} Jean-Daniel Bancal, Nicolas Gisin, Yeong-Cherng Liang, and Stefano Pironio; Device-Independent Witnesses of Genuine Multipartite Entanglement; \href{https://journals.aps.org/prl/abstract/10.1103/PhysRevLett.106.250404}{Phys. Rev. Lett. \textbf{106}, 250404 (2011)}

\bibitem{Bancalbook2014} Jean-Daniel Bancal; Device-Independent Witnesses of Genuine Multipartite Entanglement; \href{https://link.springer.com/book/10.1007/978-3-319-01183-7}{2014}

\bibitem{Liangprl2014} Y-C Liang, D. Rosset, J-D Bancal, G. Pütz, TJ Barnea, and N. Gisin; Family of Bell-like Inequalities as Device-Independent Witnesses for Entanglement Depth; \href{https://journals.aps.org/prl/abstract/10.1103/PhysRevLett.114.190401}{Phys. Rev. Lett. \textbf{114}, 190401(2014) }

\bibitem{Acinquantum2018} Randomness versus nonlocality in the Mermin-Bell experiment with three parties; Erik Woodhead, Boris Bourdoncle, Antonio Acín; \href{https://quantum-journal.org/papers/q-2018-08-17-82/}{	Quantum 2, 82 (2018)}

\bibitem{Aolitaprl2012} L. Aolita, R. Gallego, A. Cabello, and A. Acín; Fully Nonlocal, Monogamous, and Random Genuinely Multipartite Quantum Correlations; \href{https://journals.aps.org/prl/abstract/10.1103/PhysRevLett.108.100401}{Phys. Rev. Lett. \textbf{108}, 100401 ,(2012)}.


\bibitem{Perale1970} P.M.Pearle, Hidden-Variable Example Based upon Data Rejection; \href{https://journals.aps.org/prd/abstract/10.1103/PhysRevD.2.1418}{Phys. Rev. D \textbf{2}, 1418 (1970)}

\bibitem{Garg1987} A. Garg and N. D. Mermin; Detector inefficiencies in the Einstein-Podolsky-Rosen experiment; \href{https://journals.aps.org/prd/abstract/10.1103/PhysRevD.35.3831}{Phys. Rev. D \textbf{35}, 3831 (1987)}
\bibitem{Larssson1998} Jan-Åke Larsson; Bell’s inequality and detector inefficiency; \href{https://journals.aps.org/pra/abstract/10.1103/PhysRevA.57.3304}{Phys. Rev. A \textbf{57}, 3304 (1998)}

\bibitem{Eberhard1993}
P.H.Ebarhard; Background level and counter efficiencies required for a loophole-free Einstein-Podolsky-Rosen experiment; \href{https://journals.aps.org/pra/abstract/10.1103/PhysRevA.47.R747}{Phys. Rev. A \textbf{47}, R747(R) (1993)}
\bibitem{Larson2007} A. Cabello and J-Å Larsson; Minimum Detection Efficiency for a Loophole-Free Atom-Photon Bell Experiment; \href{https://journals.aps.org/prl/abstract/10.1103/PhysRevLett.98.220402}{Phys. Rev. Lett. \textbf{98}, 220402 (2007)}

\bibitem{Brunner2007prl} N. Brunner, N. Gisin, V. Scarani, and C. Simon; Detection Loophole in Asymmetric Bell Experiments; \href{https://journals.aps.org/prl/abstract/10.1103/PhysRevLett.98.220403}{Phys. Rev. Lett. \textbf{98}, 220403(2007)}

\bibitem{Monroe2004} B. B. Blinov, D. L. Moehring, L.- M. Duan and  C. Monroe; Observation of entanglement between a single trapped atom and a single photon,\href{https://www.nature.com/articles/nature02377}{Nature. \textbf{428}, 153–157 (2004)} 
\bibitem{Hanson2015} Hanson et. al; Loophole-free Bell inequality violation using electron spins separated by 1.3 kilometres, \href{https://www.nature.com/articles/nature15759}{Nature, \textbf{526}, 682(2015)}

\bibitem{Zeilinger2015} Zeilinger et. al; Significant-Loophole-Free Test of Bell’s Theorem with Entangled Photons, \href{https://journals.aps.org/prl/abstract/10.1103/PhysRevLett.115.250401}{Phys. Rev. Lett. 115, 250401(2015) }
\bibitem{Shalm15} Lynden K. Shalm et al; Strong Loophole-Free Test of Local Realism, \href{https://journals.aps.org/prl/abstract/10.1103/PhysRevLett.115.250402}{Phys. Rev. Lett. \textbf{115}, 250402(2015)}
\bibitem{Wallraff2023} Wallraff. et. al; Loophole-free Bell inequality violation with superconducting circuits, \href{https://www.nature.com/articles/s41586-023-05885-0}{Nature, \textbf{617}, 265(2023)}
\bibitem{Semitecolos2001} J.-Å. Larsson and J. Semitecolos; Strict detector-efficiency bounds for n-site Clauser-Horne inequalities; \href{https://journals.aps.org/pra/abstract/10.1103/PhysRevA.63.022117}{Phys. Rev. A \textbf{63}, 022117 (2001)}
\bibitem{Villanueva2008} A. Cabello, D. Rodríguez, and I. Villanueva; Necessary and Sufficient Detection Efficiency for the Mermin Inequalities, \href{https://journals.aps.org/prl/abstract/10.1103/PhysRevLett.101.120402}{Phys. Rev. Lett. \textbf{101}, 120402 (2008)}







\bibitem{Hong2012} Y. Xiang, H-X Wang, and F-Y. Hong; Detection efficiency in the loophole-free violation of Svetlichny's inequality, \href{https://journals.aps.org/pra/abstract/10.1103/PhysRevA.86.034102}{Phys. Rev. A \textbf{86}, 034102 (2012)}

\bibitem{Smerzi2012pra} Valentin Gebhart and Augusto Smerzi; Coincidence postselection for genuine multipartite nonlocality: Causal diagrams and threshold efficiencies, \href{https://journals.aps.org/pra/abstract/10.1103/PhysRevA.106.062202}{Phys. Rev. A \textbf{106}, 062202(2022)}

  \bibitem{Banik2023} S. B. Ghosh, S. Roy Chowdhury, G. Kar, A. Roy, T. Guha, M. Banik; Quantum Nonlocality: Multi-copy Resource Inter-convertibility {\& }Their Asymptotic Inequivalence,
\href{https://arxiv.org/abs/2310.16386}{arXiv:2310.16386v1}.


\bibitem{Jeong2014} Hyunseok Jeong, Alessandro Zavatta, Minsu Kang, Seung-Woo Lee, Luca S. Costanzo, Samuele Grandi, Timothy C. Ralph, Marco Bellini; Generation of hybrid entanglement of light, \href{https://doi.org/10.1038/nphoton.2014.136}{nphoton., 8, 564–569 (2014)}

\bibitem{Wen2021} Jianming Wen, Irina Novikova, Chen Qian, Chuanwei Zhang, hengwang Du; Hybrid Entanglement between Optical Discrete Polarizations and Continuous Quadrature Variables, \href{https://doi.org/10.3390/photonics8120552}{Photonics 2021, 8(12), 552}

\bibitem{Feist2022} Armin Feist, Guanhao Huang, Germaine Arend, Yujia Yang, Jan-Wilke Henke, Arslan Sajid Raja, F. Jasmin Kappert, Rui Ning Wang, Hugo Lourenço-Martins, Zheru Qiu, Junqiu Liu, Ofer Kfir, Tobias J. Kippenberg, Claus Ropers; Cavity-mediated electron-photon pairs, \href{https://www.science.org/doi/10.1126/science.abo5037}{Science377,777-780(2022)}

\bibitem{He2022} Mingjian He, Robert Malaney; R. Teleportation of hybrid entangled states with continuous-variable entanglement, \href{https://doi.org/10.1038/s41598-022-21283-4}{Sci Rep 12, 17169 (2022)}



\end{thebibliography}
\end{document}